\documentclass[%
 reprint,
 amsmath,amssymb,
 aps,
]{revtex4-2}

\usepackage{graphicx}
\usepackage{dcolumn}
\usepackage{bm}
\usepackage{amsmath}
\usepackage{amsfonts}
\usepackage{amsthm}
\newtheorem{theorem}{Theorem}

\newtheorem{lemma}[theorem]{Lemma}

\newtheorem{definition}[theorem]{Definition}

\usepackage[hidelinks]{hyperref}
\usepackage{physics}
\usepackage{color}
\usepackage{float}

\begin{document}

\preprint{APS/123-QED}

\title{No-Go Theorem for Gaussian Quantum Repeaters from Fractional Extendibility}

\author{Rabsan Galib Ahmed$^{1,2}$}
 \email{rgahmed@uwaterloo.ca}
\author{Graeme Smith$^{1,2}$}%
 \email{graeme.smith@uwaterloo.ca}
\affiliation{$^1$Institute for Quantum Computing, University of Waterloo, Ontario N2L 3G1, Canada.\\
$^2$Department of Applied Mathematics, University of Waterloo, Ontario N2L 3G1, Canada.}%




\date{\today}

\begin{abstract}
         Photon loss in optical channels fundamentally limits long-range reliable quantum communication. A standard approach to overcoming this limitation is the use of quantum repeater nodes, which typically perform experimentally demanding non-Gaussian operations. However, whether  Gaussian repeater protocols can enhance quantum communication rates over bosonic attenuation channels has remained open. In this work, we prove a no-go theorem for Gaussian quantum repeaters in a quantum network. Specifically, we show that any repeater chain composed of Gaussian operations, homodyne measurements, and arbitrary classical communication cannot enhance the quantum capacity of a pure-loss attenuation channel beyond that achievable by direct transmission. Our proof introduces a generalisation of $k$-extendibility to a notion of \textit{fractional extendibility} for Gaussian states and establishes some of its useful properties, thereby providing a powerful framework for analysing Gaussian quantum networks.
\end{abstract}

\maketitle


Practical realization of quantum technologies such as quantum key distribution~\cite{Bennett2014} and quantum networks~\cite{Kimble2008} requires reliable long-distance transmission of quantum states. The leading physical platform for such communication is provided by optical bosonic channels, including optical fibers and free-space links. However, quantum communication over optical channels is fundamentally limited by photon loss, which induces exponential attenuation with distance and severely suppresses the rate of transmissible quantum information~\cite{Takeoka2014,Pirandola2017}. Overcoming this loss-induced decay, which determines the ultimate rate at which quantum information, entanglement, or secret key can be distributed across optical networks, is a central challenge of quantum networking.

Bosonic attenuation channels, which model photon-loss, are characterized by their transmissivity~$\eta \in [0,1]$, which quantifies the fraction of input energy transmitted through the channel~\cite{Weedbrook2012}. The ultimate rate at which quantum information can be reliably transmitted over arbitrarily many uses of a quantum channel is defined as its quantum capacity. For the pure-loss bosonic attenuation channel $\mathcal{A}_\eta$, the quantum capacity is given by $\mathcal{Q}(\mathcal{A}_{\eta}) = \max\{0,\log_2\frac{\eta}{1-\eta}\}$, under an unconstrained input-energy assumption~\cite{Giovannetti2014,Holevo2001,Wolf2007}. Hence, in a standard optical fiber, operating at a loss rate $0.2$ dB/km~\cite{Pirandola2019}, it is impossible to perform one-way quantum communication at a positive rate beyond $15$ km.

The vanishing quantum capacity of attenuation channels with transmissivity $\eta \leq 1/2$ can be immediately understood by noting that such a channel is $2$-extendible. A bipartite quantum state $\rho_{AB}$ (and its isomorphic channel) is said to be $k$-extendible~\cite{Li_2018} with respect to subsystem $B$ if there exists a quantum state $\rho_{AB_1B_2\dots B_k}$ such that $\rho_{AB} = \rho_{AB_j} = \Tr_{B^k\setminus B_j }[\rho_{AB_1B_2\dots B_k}]$. If the receiver could reliably decode any arbitrary quantum state sent over these noisy channels, so could the environment. This is in contradiction with the \textit{no-cloning} theorem~\cite{Wootters1982}, leading to zero quantum capacity of these channels.


The primary approach to overcoming this difficulty in optical communication is the use of quantum repeater stations inserted between successive segments of attenuation channels. In principle, such repeaters can suppress the exponential decay of transmissivity with distance by repeatedly decoding, correcting, and re-encoding the transmitted quantum information, thereby enabling communication rates beyond those achievable by direct transmission alone~\cite{Briegel1998,Kimble2008,Pirandola2017,Azuma2023}. 

In optical platforms, however, experimentally accessible operations are predominantly Gaussian.  These include linear optics, squeezers, homodyne measurements, feedforward, and Gaussian ancillae. A major obstacle is the no-go theorem for Gaussian quantum error correction, which states that Gaussian errors (including attenuation) on Gaussian states cannot be corrected using only Gaussian encoding and Gaussian recovery operations~\cite{Niset2009}. Closely related no-go results also exclude Gaussian entanglement distillation using Gaussian local operations and classical communication alone~\cite{Eisert2002,Fiurasek2002,GC2002}. Therefore, it seems that repeater architecture for attenuation channels would need to incorporate non-Gaussian resources or operations~\cite{Azuma2023}, which are usually experimentally demanding and comparatively resource intensive~\cite{Weedbrook2012,Ra_2019}. However, the fundamental question remains: \textit{Can Gaussian repeater protocols enhance quantum communication rates over bosonic attenuation channels?}

In Ref.~\cite{Namiki2014}, it was shown that Gaussian regenerative stations based on ancillary Gaussian-state preparation, joint Gaussian unitary interactions, and unconditional tracing over ancillary modes cannot improve the capacity of attenuation channels. Crucially, that framework excludes an essential component of quantum repeaters: \textit{measurements}. A Gaussian quantum repeater protocol should more generally include adaptive Gaussian operations involving homodyne or heterodyne measurements on ancillary modes, feedforward-conditioned Gaussian transformations, and forward transmission of classical side information to subsequent repeater nodes. As remarked in Ref.~\cite{Namiki2014}, it therefore remained open whether such fully general Gaussian repeater protocols could enhance quantum communication rates over attenuation channels.

In this work, we prove a no-go theorem for Gaussian quantum repeaters in a quantum network. More specifically, we show that the quantum capacity of a pure-loss attenuation channel cannot be enhanced by inserting repeater nodes capable of performing Gaussian operations including measurements, assisted by arbitrary classical communication among the repeater nodes. To prove this, we introduce a generalisation of Gaussian $k$-extendibility to \textit{fractional extendibility}. We show that it is monotonic under arbitrary Gaussian operations, exhibits a simple composition rule, and provides bounds on the rate of quantum communication. We expect the framework established in this paper to have applications in analysing further complicated Gaussian quantum network.      

\textit{Preliminaries ---} An $N$-mode bosonic system is defined by $N$ pairs of \textit{canonical observables}, $\mathbf{R}:=(x_1,p_1,\dots x_N,p_N)^T$, satisfying the canonical commutation relation (CCR), $[R_j,R_k] = i\Omega_{jk}$, where $\Omega := \begin{pmatrix}
    0 & 1\\
    -1 & 0
\end{pmatrix}^{\oplus N}$. 
Given a quantum state $\rho$ of an $N$-mode bosonic systems, the first moment and the quantum covariance matrix (QCM) are defined by $d_j = \Tr(\rho R_j)$ and $V_{jk} = \Tr[(R_jR_k+R_kR_j)\rho] - d_jd_k$ respectively. The uncertainty principle implies that every QCM must satisfy, $V\geq i\Omega$. A particularly important class of bosonic states, called the \textit{Gaussian states} are defined to be thermal states of Hamiltonians that are quadratic in $R_j$'s with eigenvalues bounded from below. Moreover, these states are entirely determined by their first moment and QCM~\cite{Simon1994,Serafini2017}.  

Quantum operations preserving Gaussianity are called Gaussian operations. Operationally, they can be realized using ancillary Gaussian states, Gaussian unitaries generated by quadratic Hamiltonians, and homodyne measurements. Important examples of Gaussian unitaries include phase shifts, squeezing transformations, beam splitters, and two-mode squeezers~\cite{Weedbrook2012}. Linear Bosonic Gaussian channels correspond to Gaussian operations where measurement outcomes are discarded~\cite{Holevo2001}. These are the channels that has been considered in~\cite{Namiki2014} as Gaussian regenerative stations.

 An important example of a linear Gaussian channel is the pure-loss attenuation channel, $\mathcal{A}_\eta$. Mathematically, it is realized by mixing the input modes with some environmental vacuum modes on a beam splitter of transmissivity $\eta$, followed by tracing out the environment~\cite{eisert2005}. Physically, it arises when a fraction $1-\eta$ of a signal is absorbed. Its action on Gaussian states is $V\mapsto \eta \;V+(1-\eta)\mathbb{I}$ and $d\mapsto \sqrt{\eta}\;d$. 

Any Gaussian operation $\mathcal{G}_{A\to B}$ admits a Gaussian Choi-Jamio\l{}kowski representation obtained from its action on tensor products of two-mode squeezed vacuum (TMSV) states~\cite{LKAW2019}. Specifically for an operation with $N$-mode input, the Choi-Jamio\l{}kowski operator is given by, $\rho^{\mathcal{G}}_{AB}(r) = \mathcal{G}_{A'\to B}(\ketbra{\psi_r}{\psi_r}_{AA'}^{\otimes N})$, where $\ket{\psi_r}:=\mathrm{sech}{(r)}\sum_{n}(\tanh(r))^n\ket{n}\ket{n}$ are called the two-mode squeezed vacuum (TMSV) states with squeezing parameter $r>0$. As TMSV states are Gaussian, so are $\rho^{\mathcal{G}}_{AB}(r)$~\footnote{If $\mathcal{G}$ is trace non-increasing, then $\rho^{\mathcal{G}}_{AB}(r)$ would be an unnormalised Gaussian state}. The action of $\mathcal{G}$ on some input Gaussian state with QCM and first moments, $V_A$ and $d_A$ respectively is given by~\cite{GC2002} 
\begin{align*}
    V_A\mapsto V'_B &= B - \tilde{C}^T (V_A+\tilde{A})^{-1} \tilde{C}\\
    d_A\mapsto d'_B &= D_B + \tilde{C}^T (V_A+\tilde{A})^{-1}(D_A+d_A),
\end{align*}
where $\Gamma_{AB}\equiv \begin{pmatrix}
    A & C\\
    C^T & B
\end{pmatrix},D\equiv \begin{pmatrix}
    D_A\\D_B
\end{pmatrix}$ are respectively the QCM and the first moments of the Choi-Jamio\l{}kowski operator, $\rho^{\mathcal{G}}_{AB}(r)$ in the limit $r\to \infty$, and we have denoted $\tilde{\Gamma}_{AB}:=  (\Lambda_A \;\oplus \mathbb{I}_{B}) \;\Gamma_{AB} \; (\Lambda_A \;\oplus \mathbb{I}_{B})$, where $\Lambda_A = \bigoplus_{i=1}^N \mathrm{diag}(1,-1)$.

Schur complement is an important mathematical tool for analysing Gaussian quantum states and dynamics~\cite{LHAW2016,LSA2018}. Consider a Hermitian matrix $M$ divided in blocks as
    \begin{align*}
        M:= \begin{pmatrix}
            A & X\\
            X^\dagger & B
        \end{pmatrix},
    \end{align*}
    in which $A$ is square and non-singular. $M$ is positive semidefinite, i.e., $M \geq 0$ if and only if the \textit{Schur complement of $A$ in $M$}, $M/A := B - X^\dagger A^{-1}X \geq 0$. Additionally $M>0$ if and only if $A>0$ and $M/A >0$~\cite{Bhatia2015}. Moreover, for any positive semidefinite matrix $M$ partitioned as above, the Schur complement $M/A := B - X^\dagger A^{-}X$ is well-defined for any generalised inverse of $A$. In particular it is well-defined if the inverse is taken on the support of $A$~\cite{Zhang2005}. 

    Consider Hermitian matrices $M,M'$ partitioned similarly as above with corresponding blocks $A$ and $A'$. Two important results involving Schur complements are~\cite{Horn1985}: (i) \textit{Monotonicity}: If $M\geq M'$, then $M/A\geq M'/A'$. (ii) \textit{Joint concavity}: If $A$ and $A'$ are positive-definite, $[(1-x)M+xM']/[(1-x)A+xA'] \geq (1-x) M/A + x M'/A'$ where $x \in [0,1]$.

 \textit{Main results ---} Suppose $\mathcal{A}_\eta$ is an $N$-mode Gaussian attenuation channel with transmissivity $\eta$. Further, suppose that forward quantum communication between two users over $\mathcal{A}_\eta$ is assisted by arbitrary finite number of intermediate repeater nodes, where each node is restricted to performing arbitrary Gaussian operation. Moreover, arbitrary classical communication is allowed among the repeater nodes. As a special case, our results apply to the scenario where only forward classical communication is allowed which is comparable to the third generation of quantum repeaters~\cite{Azuma2023}.

 Specifically, we analyse the situation where at first the repeater nodes locally prepare some bipartite Gaussian states, possibly assisted by some classical communication among themselves. Then each node shares one part to the next node through the intermediate attenuation channel. Afterwards they perform arbitrary multipartite Gaussian local operation and classical communication (GLOCC) protocol~\cite{GC2002,Eisert2002,Fiurasek2002} to output a bipartite entangled state between the first and the last node. The first node then teleports the state it receives from the sender to the last node using that entanglement. The last node then forwards the teleported state to the receiver through the last attenuation channel (See Fig.~\ref{fig:glocc-choi state}). 

\begin{figure}
    \centering
    \includegraphics[width=1\linewidth]{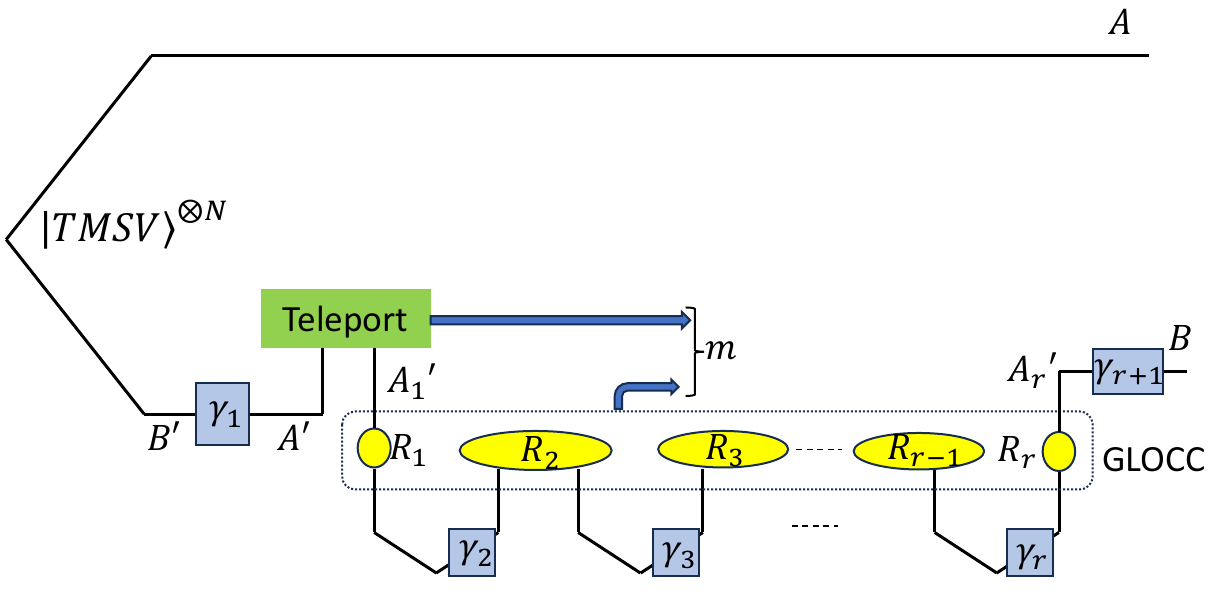}
    \caption{Choi state of the repeater chain, $\mathcal{R}_{\eta}$ corresponding to some message string $m$. $\gamma_i$'s denote attenuation parameters of the intermediate pure loss channels.}
    \label{fig:glocc-choi state}
\end{figure}
 
 The equivalent channel from the sender ($A$)  to the receiver ($B$) is denoted by $\mathcal{R}_\eta : A\to BM$, where $M$ is a classical register containing all the classical information produced in the entire protocol.  
 
  The main result of this paper is the following.

\begin{theorem}[No-Go theorem for Gaussian repeater chain]\label{thm:main_1}
    The quantum capacity of a Gaussian repeater chain, $\mathcal{R}_\eta$ is upper bounded by that of the attenuation channel, $\mathcal{A}_\eta$, i.e., $\mathcal{Q}(\mathcal{R}_\eta) \leq \mathcal{Q}(\mathcal{A}_\eta)$.
\end{theorem}

The proof follows from a generalisation of the Gaussian $k$-extendable states~\cite{LKAW2019,Rajarama_Bhat_2017}. In~\cite{LKAW2019}, the necessary and sufficient condition for $k$-extendibility of a bipartite bosonic Gaussian state has been obtained in terms of its covariance matrix. Here we introduce the following new definition.
\begin{definition}[Fractional extendibility]
    A bipartite bosonic Gaussian state on $A|B$ with covariance matrix $V_{AB}$ is \textbf{$1/\gamma$-extendable} if there exists a covariance matrix on $B$, $\Delta_B\geq i\Omega_{B}$ such that, for $\gamma \in [0,1]$
    \begin{align}\label{eq:fract-ext}
        V_{AB} \geq i\Omega_A \oplus [(1-\gamma)\Delta_B +i\gamma \Omega_{B}].
    \end{align}
\end{definition}
With $\gamma = 1/k$ for integers $k\geq 1$, we recover the standard matrix inequality condition for $k$-extendibility in \cite{LKAW2019}. We postpone the physical meaning of Gaussian fractional extendibility till Lemma~\ref{lem:equiv-to-att}. Some crucial results regarding fractional extendibility that we need are stated and proven below.
\begin{lemma}[Composition with attenuation]\label{lem:comp-att}
    Suppose a QCM, $V_{AB}$ is $1/\gamma$-extendable. After the application of an attenuation channel, $\mathcal{A}_{\lambda}:B\to B$, the QCM, $V'_{AB}$ is $1/(\lambda\gamma)$-extendable. 
\end{lemma}
\begin{proof}
    Given  $V_{AB}$ satisfies Inequality~\eqref{eq:fract-ext}, conjugating both sides with $\mathbb{I}_A\oplus \sqrt{\lambda}\mathbb{I}_B$ and adding the noise $0_A\oplus (1-\lambda)\mathbb{I}_B$ on both sides, we obtain, $V'_{AB} \geq i\Omega_A \oplus R_B$. Here, 
    \begin{align*}
        R_B &= \lambda(1-\gamma) \Delta_B + (1-\lambda)\mathbb{I}_B +i\lambda\gamma \Omega_{B}\\
        &= (1-\lambda\gamma)\Delta'_B + i\lambda\gamma \Omega_{B}.
    \end{align*}
    It follows that $(1-\lambda\gamma)\Delta'_B = \lambda(1-\gamma) \Delta_B + (1-\lambda)\mathbb{I}_B  \geq  \lambda(1-\gamma) i\Omega_B+ (1-\lambda)i\Omega_B = (1-\lambda\gamma)i\Omega_B$ which implies that, $\Delta'_B\geq i\Omega_B$. Hence, $V'_{AB}$ is $1/(\lambda\gamma)$-extendable.
\end{proof}

\begin{lemma}[Monotonicity]\label{lem:monotonicity}
    Fractional extendibility is monotonic under Gaussian operations.
\end{lemma}
\begin{proof}
    Suppose a bipartite $1/\gamma$-extendable state with QCM, $V_{AB}$ and the Gaussian operation $\mathcal{G}:B\to B'$ has an associated Choi operator with QCM, $\Gamma_{BB'}$. The QCM of the final state on $AB'$ is given by the Schur complement~\cite{GC2002}\footnote{Technically, we should write as  $$ [(V_{AB} \oplus 0_{B'}) + (0_A\oplus\tilde{\Gamma}_{BB'})]/(V_{B}+\tilde{\Gamma}_B).$$ However, we omit writing the zero blocks for readability and it should be understood from the labelled indices.}, \  $V'_{AB'} = (V_{AB} + \tilde{\Gamma}_{BB'})/(V_B+\tilde{\Gamma}_{B})$. Using Inequality~\eqref{eq:fract-ext} and Schur complement monotonicity \cite{LHAW2016}, we obtain, $V'_{AB'} \geq i\Omega_A \oplus R_{B'}$, where  $$ R_{B'} = [(1-\gamma)\Delta_B +i\gamma \Omega_{B}+ \tilde{\Gamma}_{BB'}]/[(1-\gamma)\Delta_B +i\gamma \Omega_{B}+ \tilde{\Gamma}_{B}].$$
    Note that $(1-\gamma)\Delta_B +i\gamma \Omega_{B}+ \tilde{\Gamma}_{BB'} = (1-\gamma)M_{BB'} +\gamma N_{BB'}$, where $M_{BB'} = \Delta_B+ \tilde{\Gamma}_{BB'}$ and $N_{BB'} = i\Omega_{B} + \tilde{\Gamma}_{BB'}$. Therefore using the joint concavity of Schur complement~\footnote{The positive definiteness of the blocks can be assured by appropriate regularisation and finally taking the limit using the fact that the set of Positive Semidefinite matrices is closed.},
    \begin{align}
        R_{B'} &= [(1-\gamma)M_{BB'} +\gamma N_{BB'}]/[(1-\gamma)M_{B} +\gamma N_{B}]\nonumber\\
        &\geq (1-\gamma)(M_{BB'}/M_{B}) + \gamma(N_{BB'}/N_{B}).
    \end{align}

    As $\Gamma_{BB'}$ is a QCM, $\tilde{\Gamma}_{BB'} \geq (-i\Omega_{B})\oplus i\Omega_{B'}$. Furthermore, using $\Delta_B\geq i\Omega_B$, it can be shown that the real symmetric matrix, $\Delta'_{B'}:= M_{BB'}/M_{B}$
    satisfies $\Delta'_{B'} \geq i\Omega_{B'}$ and $N_{BB'}/N_{B}\geq i\Omega_{B'}$. Combining, we have $R_{B'}
        \geq (1-\gamma)\Delta'_{B'} + i\gamma \Omega_{B'}$. Finally, recalling that the QCM of the state after the Gaussian operation satisfies, $V'_{AB'} \geq i\Omega_A \oplus R_{B'}$, we conclude the proof.
\end{proof}
A special remark is in order when the Gaussian operation involves measurement. In such case, the preceding lemma is to be understood as -- for each measurement outcome, the state on $AB'$ is at least as extendable as the state on $AB$. This is due to the fact that the QCM of the post-measurement state is independent of the measurement outcome~\cite{Serafini2017}.

The next result shows that any $1/\gamma$-extendable state on $AB$ can be obtained from some other bipartite Gaussian state by applying the attenuation channel, $\mathcal{A}_{\gamma}$, to the $B$ system followed by some Gaussian unitary.

\begin{lemma}\label{lem:equiv-to-att}
    For every $1/\gamma$-extendable Gaussian state, $\rho_{AB}$, there exists another Gaussian state $\sigma_{AB}$ such that $\rho_{AB} = \mathrm{id}_A\otimes (\mathcal{S}\circ \mathcal{A}_{\gamma}) (\sigma_{AB})$, where $\mathcal{S}$ denotes conjugation by some Gaussian unitary and $\mathcal{A}_{\gamma}$ denotes the attenuation channel with transmissivity $\gamma$.
\end{lemma}
\begin{proof}
    We can set the first moments of $B$-system quadratures in $\rho_{AB}$ to zero without loss of generality, because a displacement operator can always be included in $\mathcal{S}$. Furthermore, $\sigma_{AB}$ can be chosen to have the same first moments of $A$-system quadratures as $\rho_{AB}$. Let the QCM for $\rho_{AB}$ be given by $V^{\rho}_{AB}$, which satisfies Inequality~\eqref{eq:fract-ext} for some QCM on $B$, $\Delta_B$. As the symplectic eigenvalues of $\Delta_B$ are greater than or equal to $1$, there exists symplectic matrix, $S$, such that, $\Delta_B \geq S \;\mathbb{I}_B\;S^T$. Therefore, we obtain, $V^\rho_{AB} \geq i\Omega_A \oplus [(1-\gamma)S \;\mathbb{I}_B\;S^T +i\gamma \Omega_{B}]$. Simple algebraic manipulations reveal that
    \begin{align}
        V^\sigma_{AB} &:= \left(\mathbb{I}_A\oplus \frac{S^{-1}}{\sqrt{\gamma}}\right)V^\rho_{AB}\left(\mathbb{I}_A\oplus \frac{(S^{-1})^T}{\sqrt{\gamma}}\right) -  \frac{1-\gamma}{\gamma}\mathbb{I}_B \nonumber\\
        &\geq i\Omega_A \oplus i\Omega_B.
    \end{align}
    Hence, $V^\sigma_{AB}$ is a valid QCM corresponding to some Gaussian state, $\sigma_{AB}$. Moreover, it is trivial to check from the expressions above that $V^{\rho}_{AB}$ is obtained from $V^\sigma_{AB}$ by applying the attenuation channel, $\mathcal{A}_{\gamma}$ followed by conjugation by the symplectic matrix, $S$, completing the proof.  
\end{proof}
The previous lemma enables us to physically interpret the fractional extendibility of Gaussian states. For $\gamma = q/p$, any $1/\gamma$-extendable Gaussian state, $\rho_{AB}$ is obtained from some other Gaussian state $\sigma_{AB}$, by \textit{symmetrically splitting} $B$ in $p$ modes using the consecutive beam-splitter unitary of an appropriate parameter and forwarding only $q$ modes to the receiver~\footnote{For an irrational $\gamma$, one can extend this interpretation by considering a sequence of rational numbers approaching $\gamma$.}. Whenever $q$ divides $p$, we obtain the notion of usual $k$-extendibility, with $k=p/q$, as the environment can then simulate $k-1$ copies of what the receiver receives. 

The next result is more specific to repeater chains with the set-up considered as follows. Let repeater nodes $R_k$ for $1\leq k\leq r-1$ prepare bipartite local Gaussian states on $A_{k+1,k}A_{k+1,k+1}$ respectively. Then each $R_k$ sends $A_{k+1,k+1}$ to $R_{k+1}$ through the attenuation channel $\mathcal{A}_{\gamma_{k+1}}$ respectively. Afterwards, the repeater nodes proceed with an arbitrary multipartite GLOCC protocol, $\mathcal{G}:(A_{k+1,k}A_{k+1,k+1})_{k=1}^{r-1}\to (A'_{k+1,k}A'_{k+1,k+1})_{k=1}^{r-1}$. As an example, the case for three repeater nodes is illustrated in Fig.~\ref{fig:GLOCC}. 
\begin{figure}
    \centering
    \includegraphics[width=0.6\linewidth]{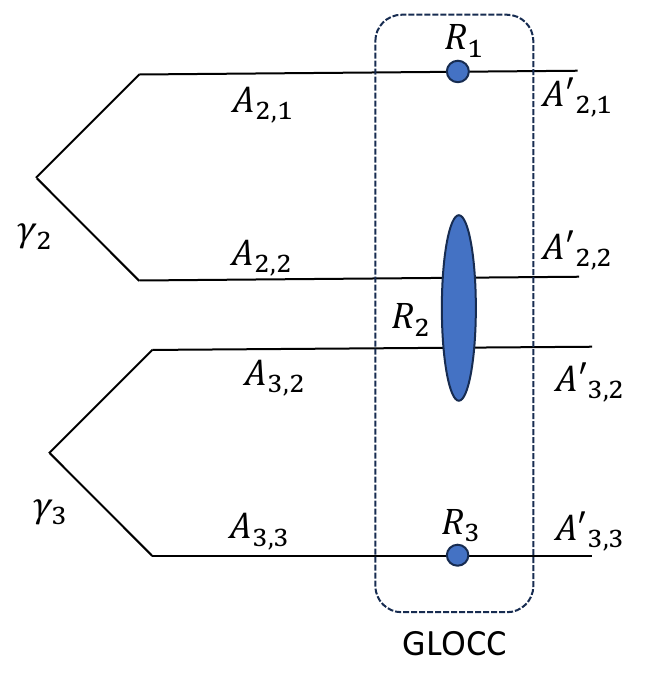}
    \caption{Each repeater node prepares a bipartite state locally and shares one part to the next node using the attenuation channel (illustrated for three nodes). Thereafter they perform an arbitrary GLOCC protocol.}
    \label{fig:GLOCC}
\end{figure}
\begin{lemma}\label{lem:glocc}
    For any Gaussian CP map separable across the repeater nodes (e.g. Fig.~\ref{fig:GLOCC}), the output state on $A'_{2,1}A_{r,r}'$ is $(\prod_{k=2}^r\gamma_k)^{-1}$-extendable.
\end{lemma}
The proof is given in Appendix~\ref{appsec:proof}. Now we are in a position to prove Theorem~\ref{thm:main_1}.

\textit{Proof of Theorem~\ref{thm:main_1} --} We prepare $N$ two-mode squeezed vacuum states, with QCM $V_{AB'}$, where $A$ and $B'$ each consists of $N$ modes. We send $B'$ through $\mathcal{A}_{\gamma_1}$. Using Lemma~\ref{lem:comp-att}, we conclude that the output bipartite state on $AA'$ is $1/\gamma_1$-extendable.

Furthermore, using Lemma~\ref{lem:glocc}, we find that corresponding to each branch of the GLOCC protocol, labelled by some index $m_1$, the bipartite state on the first and the last node, denoted by $\tilde{\Phi}_{m_1}$ on $A_1'A_r'$, is $1/(\gamma_2\gamma_3\dots\gamma_r)$-extendable. Now the first node teleports the state of the system $A'$ through $\tilde{\Phi}_{m_1}$ to the system, $A'_r$ that the last node possesses (See Fig.~\ref{fig:glocc-choi state}).

Now Lemma~\ref{lem:equiv-to-att} implies that there exists a Gaussian state $\Psi^{m_1}_{A_1'A_r'}$ on $A_1'A_r'$ and some Gaussian unitary channel, $\mathcal{S}_{m_1}$ such that $\tilde{\Phi}_{m_1} = \mathrm{id}_{A_1'}\otimes (\mathcal{S}_{m_1}\circ \mathcal{A}_{\gamma_2\dots\gamma_r})(\Psi^{m_1}_{A_1'A_r'})$. Furthermore as all the states and channels are Gaussian, all post-teleported states are equivalent to each other up to local displacements~\cite{GC2002}. Therefore, teleporting $A'$ through $\tilde{\Phi}_{m_1}$ is equivalent to the following sequence: \textit{(i) Teleportation of $A'$ through $\Psi^{m_1}_{A_1'A_r'}$}: The output state $AA'_r$ is $1/\gamma_1$-extendable (Lemma~\ref{lem:monotonicity}). \textit{(ii) Application of an attenuation channel, $\mathcal{A}_{\gamma_2\dots\gamma_r}$}: The output state $AA'_r$ is $(\prod_{k=1}^r\gamma_k)^{-1}$-extendable (Lemma~\ref{lem:comp-att}). \textit{(iii) Gaussian Unitary channel, $\mathcal{S}_{m_1}$}:  The output state $AA'_r$ is $(\prod_{k=1}^r\gamma_k)^{-1}$-extendable (Lemma~\ref{lem:monotonicity}).

Finally the system $A'_r$ is sent through the last attenuation channel, $\mathcal{A}_{\gamma_{r+1}}$ to the receiver $B$ along with the label $m_1$ and the teleportation outcome, collectively indexed by $\vec{m}$, on a classical register $M$. Therefore, for every message string, $\vec{m}$, the Gaussian state on $AB$ is $(\prod_{k=1}^{r+1}\gamma_k)^{-1} = \eta^{-1}$-extendable (Lemma~\ref{lem:comp-att}). 

 Therefore, using Lemma~\ref{lem:equiv-to-att} we state that the Choi state of the repeater chain corresponding to any message string $\vec{m}$ can be written as $\rho^{\vec{m}}_{AB} = \mathrm{id}_A\otimes (\mathcal{S}_{\vec{m}}\circ \mathcal{A}_{\eta}) (\sigma^{\vec{m}}_{AB})$, for some Gaussian state $\sigma^{\vec{m}}_{AB}$. Therefore, the repeater chain, $\mathcal{R}_\eta:B'\to BM$ corresponds to the Choi state $\sum_{\vec{m}}p_{\vec{m}} \rho^{\vec{m}}_{AB} \otimes \ketbra{\vec{m}}{\vec{m}} = \sum_{\vec{m}}  \mathrm{id}_A\otimes (\mathcal{S}_{\vec{m}}\circ \mathcal{A}_{\eta}) (p_{\vec{m}}\sigma^{\vec{m}}_{AB}) \otimes \ketbra{\vec{m}}{\vec{m}}$, where $p_{\vec{m}}$ is some probability distribution over the message strings. Invoking Choi-Jamio\l kowski isomorphism between the positive semidefinite operators $p_{\vec{m}}\sigma^{\vec{m}}_{AB}$ and CP maps $\mathcal{N}_{\vec{m}}:B'\to B$ such that $\sum_{\vec{m}}\mathcal{N}_{\vec{m}}$ is trace-preserving, we can write $\mathcal{R}_\eta$ as a composition of three channels: $
    \mathcal{R}_\eta = \mathcal{U} \circ (\mathcal{A}_\eta\otimes \tilde{\mathrm{id}}_M)\circ \mathcal{V}$, where, $\mathcal{V}:B'\to BM$ is given by $\mathcal{V}(\rho)=\sum_{\vec{m}}\mathcal{N}_{\vec{m}}(\rho)\otimes \ketbra{\vec{m}}{\vec{m}}$, $\mathcal{U}:BM\to BM$ is given by $\mathcal{U}(\rho)=\sum_{\vec{m}}(S_{\vec{m}}\otimes \ketbra{\vec{m}}{\vec{m}})\rho(S^\dagger_{\vec{m}}\otimes \ketbra{\vec{m}}{\vec{m}})$, and $\tilde{\mathrm{id}}_M$ is the completely dephasing channel on $M$. Using the bottleneck inequality~\cite{RMG2018}: $\mathcal{Q}(\mathcal{N}\circ \mathcal{M}) \leq \min (\mathcal{Q}(\mathcal{N}),\mathcal{Q}(\mathcal{M}))$, we can write that
\begin{align}
    \mathcal{Q}(\mathcal{R}_{\eta}) \leq \mathcal{Q}(\mathcal{A}_{\eta} \otimes \tilde{\mathrm{id}}_M) = \mathcal{Q}(\mathcal{A}_{\eta}).
\end{align}
The last equality follows from the fact that $\tilde{\mathrm{id}}_M$ is an entanglement-breaking channel, which cannot increase the quantum capacity of a channel~\cite{Barnum1998,BKN2000,Wilde2013}. Hence, we have proven the main result of the paper. \qed

\textit{Concluding remarks ---} In this work, we have established the fundamental necessity of non-Gaussian operations at the repeater station for achieving higher rate of quantum communication than that of direct transmission. We hope that the techniques and results of this paper can be further generalised to more complicated Gaussian repeater networks revealing the associated limits. A crucial component of our method comprises of introducing the idea of fractional extendibility for Gaussian states. However, an appropriate discrete variable analogue remains to be explored in a future work.

\textit{Acknowledgements ---} We thank Yunkai Wang for helpful discussions. GS is
supported under NSERC-NSF alliance grant ALLRP-586858-2023 and NSERC Discovery grant
RGPIN-2025-02094. RGA acknowledges the support of Institue for Quantum Computing and Mike and Ophelia Lazaridis Graduate Fellowship.
\nocite{*}

\bibliography{apssamp}
\appendix
\section{Proof of Lemma~\ref{lem:glocc}}\label{appsec:proof}
In this appendix, we include the proof of Lemma~\ref{lem:glocc}.
\begin{proof}
    After the repeater nodes prepare the bipartite Gaussian states locally and share one part forward, the QCM of the $r$-partite Gaussian state is given by $V_A = \bigoplus_{k=2}^r V_k$. Using Lemma~\ref{lem:comp-att}, for each $k$, we have $V_k \geq i\Omega_{A_{k,k-1}} \oplus [(1-\gamma_k)\Delta_{k,k} + i\Omega_{A_{k,k}}]$ for some QCMs $\Delta_{k,k}$. Furthermore, as the Gaussian CP map is separable across the nodes $R_1,R_2,\dots,R_r$ as in Fig.~\ref{fig:GLOCC}, the QCM of its associated Choi operator satisfies $\Gamma_{AA'}\geq \Gamma^{\text{local}}_{AA'}:=\bigoplus_{k=1}^r \Gamma_{(AA')_k}$, for some QCMs $\Gamma_{(AA')_k}$~\cite{Werner2001}. Here $(AA')_k$ denotes the joint system local to repeater node, $k$, i.e.  $(AA')_k\equiv A_{k-1,k}A_{k,k}A'_{k-1,k}A'_{k,k}$. Therefore, we can write
    \begin{align}
        &V_A +  \tilde{\Gamma}_{AA'}\nonumber\\
        &\geq \;i\Omega_{A_{2,1}}\oplus [(1-\gamma_2)\Delta_2 + i\gamma_2\Omega_{A_{2,2}}]\oplus \left(\bigoplus_{k=3}^r V_k\right) +\tilde{\Gamma}^{\text{local}}_{AA'}\nonumber\\
        &= \;(1-\gamma_2)\underbrace{\left[i\Omega_{A_{2,1}}\oplus \left(\Delta_2\oplus \left(\bigoplus_{k=3}^r V_k\right)\right)+\tilde{\Gamma}^{\text{local}}_{AA'}\right]}_{M^{(1)}_{AA'}} \nonumber\\
        &+ \gamma_2 \underbrace{\left[i\Omega_{A_{2,1}A_{2,2}}\oplus \left(\bigoplus_{k=3}^r V_k\right)+\tilde{\Gamma}^{\text{local}}_{AA'}\right]}_{N^{(1)}_{AA'}}.\label{eq:M,N}
    \end{align}

    The QCM of the output of the CP map is given by 
    \begin{align*}
        V_{A'}' = &\;(V_A+\tilde{\Gamma}_{AA'})/(V_A+\tilde{\Gamma}_A)\nonumber\\
    {\geq} &\;[(1-\gamma_2)M^{(1)}_{AA'}+\gamma_2N^{(1)}_{AA'}]/[(1-\gamma_2)M^{(1)}_A+\gamma_2N^{(1)}_A]\\
    {\geq} &\;(1-\gamma_2)(M^{(1)}_{AA'}/M^{(1)}_A) + \gamma_2 (N^{(1)}_{AA'}/N^{(1)}_A).
    \end{align*}
    The second and the third line follow from the monotonicity and the joint concavity of Schur complement, respectively. Now, using the block diagonal structure of $\Gamma^{\text{local}}_{AA'}$ and noting that $\Delta_2\oplus \left(\bigoplus_{k=3}^r V_k\right)$ is a QCM, we can write $M^{(1)}_{AA'}/M^{(1)}_A \geq i\Omega_{A'_{2,1}} \oplus \hat{V}_{1}$, for some QCM $\hat{V}_{1}$ on the output systems that nodes $R_2,\dots,R_r$ hold.

    For the second term, involving $N^{(1)}_{AA'}$, we can perform a decomposition very similar to Eq.~\eqref{eq:M,N}:
    \begin{align}
        &N^{(1)}_{AA'}\nonumber\\ &= \left[i\Omega_{A_{2,1}A_{2,2}}\oplus \left(\bigoplus_{k=3}^r V_k\right)\right] + \tilde{\Gamma}^{\text{local}}_{AA'}\nonumber\\
        &\geq (1-\gamma_3)\underbrace{\left[i\Omega_{A_{2,1}(A)_2}\oplus \left(\Delta_{3,3}\oplus \left(\bigoplus_{k=4}^r V_k\right)\right) + \tilde{\Gamma}^{\text{local}}_{AA'}\right] }_{M^{(2)}_{AA'}}\nonumber\\
        &+ \gamma_3 \underbrace{\left[i\Omega_{A_{2,1}(A)_2A_{3,3}}\oplus \left(\bigoplus_{k=4}^r V_k\right)+ \tilde{\Gamma}^{\text{local}}_{AA'}\right]}_{N^{(2)}_{AA'}}.
    \end{align}
    Again using monotonicity and joint concavity of Schur complement, we can write 
    \begin{align*}
        N^{(1)}_{AA'}/N^{(1)}_A \geq (1-\gamma_3)(M^{(2)}_{AA'}/M^{(2)}_{A}) + \gamma_3 (N^{(2)}_{AA'}/N^{(2)}_{A}).
    \end{align*}
    Using the same argument as for $M^{(1)}_{AA'}$, we obtain $(M^{(2)}_{AA'}/M^{(2)}_{A})\geq i\Omega_{A'_{2,1}(A')_2} \oplus \hat{V}_2$, for some QCM $\hat{V}_2$ on the output systems that nodes $R_3,\dots,R_r$ hold. For the second term, we again proceed with a similar decomposition and so on. At last, we would obtain
    \begin{align}
        V'_{A'} \geq &(1-\gamma_2) \left(i\Omega_{A'_{2,1}}\oplus \hat{V}_{1}\right) \nonumber\\
        +&\sum_{k=2}^{r-1}\gamma_2\dots\gamma_k(1-\gamma_{k+1})\left(i\Omega_{A'_{2,1}(A')_2\dots(A')_k} \oplus \hat{V}_k\right)\nonumber\\
        +&\gamma_2\dots \gamma_r i\Omega_{A'_{2,1}(A')_2\dots(A')_r},
    \end{align}
    where $\hat{V}_k$ are QCMs on output systems that nodes $R_{k+1},\dots,R_r$ hold. Finally, tracing out all output systems except for the ones held by the first and the last node, the inequality above yields
    \begin{align}
        V'_{1r} \geq i\Omega_{A'_{2,1}}\oplus \left[ \left(1-\prod_{k=2}^r\gamma_k\right)\hat{V}' + \prod_{k=2}^r\gamma_k i\Omega_{A'_{r,r}}\right].
    \end{align}
    Here 
    \begin{align*}
        \hat{V}'
        = &(1-\prod_{k=2}^r\gamma_k)^{-1}[(1-\gamma_2)[\hat{V}_1]_r + \sum_{k=2}^{r-1}\gamma_2\dots\gamma_k(1-\gamma_{k+1})[\hat{V}_k]_r]\nonumber\\ \geq &i\Omega_{A'_{r,r}},
    \end{align*}~\footnote{By $[\Delta]_r$, we mean the submatrix of the QCM $\Delta$ corresponding to the system $A'_{r,r}$}. Thus we prove the claim. 
\end{proof}
\end{document}